\documentclass[11pt,a4paper]{article}
\usepackage[T1]{fontenc}
\usepackage[utf8]{inputenc}
\usepackage{lmodern}
\usepackage{amsmath,amssymb,amsthm,mathtools}
\usepackage[margin=27mm]{geometry}
\usepackage{microtype}
\usepackage{xcolor}
\usepackage[colorlinks=true,linkcolor=blue!45!black,citecolor=blue!45!black,urlcolor=blue!45!black]{hyperref}
\hypersetup{pdftitle={The Hellinger Conjecture for Antipodal Boolean Functions},pdfauthor={Vu Khac Ky},pdfsubject={Hellinger affinity, Boolean functions, and spectral matching}}
\numberwithin{equation}{section}
\newtheorem{theorem}{Theorem}[section]
\newtheorem{proposition}[theorem]{Proposition}
\newtheorem{lemma}[theorem]{Lemma}
\newtheorem{corollary}[theorem]{Corollary}
\theoremstyle{remark}

\DeclareMathOperator{\Tr}{Tr}
\DeclareMathOperator{\diag}{diag}
\DeclareMathOperator{\sgn}{sgn}
\DeclareMathOperator{\TV}{TV}
\newcommand{\E}{\mathbb E}
\newcommand{\R}{\mathbb R}
\newcommand{\Q}{\{-1,1\}}
\newcommand{\cZ}{\mathcal Z}
\newcommand{\one}{\mathbf 1}
\newcommand{\norm}[1]{\left\lVert#1\right\rVert}
\allowdisplaybreaks[1]

\title{The Hellinger Conjecture for\\Antipodal Boolean Functions}
\author{Vu Khac Ky\thanks{FPT University, Vietnam. Email: \href{mailto:kyvk2@fe.edu.vn}{\texttt{kyvk2@fe.edu.vn}}.}}
\date{October 2026}

\begin{document}
\maketitle

\begin{abstract}
The balanced Hellinger conjecture compares the Bhattacharyya coefficient
of two output distributions induced by a balanced Boolean function under
binary symmetric noise. We prove that a single coordinate minimizes this
coefficient among antipodal Boolean functions, whose values change sign
when all input signs are reversed, at every noise level.
The proof uses a Gram matrix of square-root conditional densities.
A nuclear-norm lower bound reduces the problem to matching the eigenvalues
on its even and odd Fourier subspaces, which multiplication by an
antipodal function exchanges. A coordinate function attains the minimum
matching cost. Coordinatewise monotonicity of the matrix bound extends
the result to arbitrary independent binary-input channels.
The same argument covers functions whose sign is reversed by flipping a
fixed subset of coordinates; concavity then extends the bound to randomized
binary functions. Applications include bounds on mutual information and
hypothesis-testing errors, linear Gaussian observations, and a weighted
edge-boundary inequality.
\end{abstract}


\section{Introduction and main results}\label{sec:intro}

A balanced Boolean function partitions the input space into two equally
likely classes. After the input passes through a noisy channel, estimating
the function value is a binary hypothesis-testing problem between the two
conditional output distributions. The balanced Hellinger conjecture asks
whether retaining one input bit minimizes the Bhattacharyya coefficient
of these distributions under binary symmetric noise. The resulting bound
on mutual information connects this testing problem to the
Courtade--Kumar conjecture~\cite[Proposition~1]{ABCKN17}.
We study Boolean functions for which a fixed reversal of input signs
exchanges the two classes.

Let $X$ be uniform on the Boolean cube $\Q^n$, where $n\ge1$ and
$[n]=\{1,\ldots,n\}$. A Boolean function $f:\Q^n\to\Q$ is
\emph{balanced} if $\E f(X)=0$, so its two values are equally likely.
It is \emph{antipodal}, or \emph{odd}, if $f(-x)=-f(x)$ for every $x$.
This implies balance by pairing each $x$ with $-x$.
A \emph{dictator} is a coordinate function $f(x)=x_i$.
Allowing either sign gives a \emph{signed dictator}, $f(x)=\pm x_i$.

For probability laws $P,Q$, their \emph{Hellinger affinity}, also called
the Bhattacharyya coefficient, is
$Z(P,Q)=\int\sqrt{pq}\,d\nu$, where $p,q$ are their densities relative
to a common finite dominating measure $\nu$, such as $P+Q$.
For discrete outputs the integral is simply a sum. The value is
independent of $\nu$ and lies in $[0,1]$; smaller values mean less overlap.

\begin{samepage}
For $0\le\rho\le1$, form $Y$ by independently flipping each bit of $X$
with probability $(1-\rho)/2$, so $\E[X_iY_i]=\rho$.
The \emph{noise operator} is the conditional expectation operator
\mbox{$T_\rho u(y)=\E[u(X)\mid Y=y]$}, defined for any real
function $u$ on the cube. Write $\psi(t)=\sqrt{1-t^2}$ for $-1\le t\le1$
and $\cZ_\rho(f)=\E_Y\psi(T_\rho f(Y))$. For balanced $f$, this is exactly
$Z(P_{Y\mid f(X)=1},P_{Y\mid f(X)=-1})$, as shown in
Section~\ref{sec:gram}. For a signed dictator its value is $\sqrt{1-\rho^2}$.
\end{samepage}

\begin{theorem}[Antipodal Hellinger inequality]\label{thm:main}
Every antipodal Boolean function satisfies
\begin{equation}\label{eq:main}
 \E_Y\sqrt{1-(T_\rho f(Y))^2}\ge\sqrt{1-\rho^2},
 \qquad 0\le\rho\le1.
\end{equation}
Equality is attained by every signed dictator.
\end{theorem}

The extension to independent channels is equally sharp. Let
$W_i(\cdot\mid\pm1)$ be the two conditional output laws, or rows, of
channel $i$, and let $Z_i=Z(W_i(\cdot\mid1),W_i(\cdot\mid-1))$ be
its Bhattacharyya parameter~\cite[Sec.~I]{Arikan09}. The product channel is
$W(\cdot\mid x)=\bigotimes_i W_i(\cdot\mid x_i)$: conditional on $X$,
the outputs are independent and $Y_i$ depends only on $X_i$.
For any channel $W$ and any Boolean $f$, put
$\cZ_W(f)=\E_Y\psi(\E[f(X)\mid Y])$. For balanced $f$, this is
the Bhattacharyya parameter of the induced binary-input channel with
input $f(X)$ and output $Y$.

\begin{theorem}[Independent channels]\label{thm:channels}
For independent binary-input channels, every antipodal Boolean function
satisfies $\cZ_W(f)\ge\min_i Z_i$. Equality is attained by a signed dictator
on a coordinate of minimum affinity.
\end{theorem}

The main results assume independent uniform input bits.
The coordinate channels may differ and have arbitrary measurable output
spaces. In particular, for binary symmetric channels with correlations
$\rho_i$, the bound is $\min_i\sqrt{1-\rho_i^2}$. The argument first treats identical
coordinate affinities, then compares unequal affinities to their minimum.

Expectations are taken under the stated joint law; in particular,
$\E f=\E[f(X)]$ for uniform $X$. Subscripts specify the variable
averaged when this matters. Entropy and mutual information use base-two
logarithms, whereas $\log$ and $\ln$ denote natural logarithms.

\subsection{Why symmetry helps}

Suppose each coordinate channel has Bhattacharyya parameter $h$.
The output distributions for inputs $x$ and $-x$ have affinity $h^n$,
which can be much smaller than the desired bound $h$.
Instead, we retain the affinities between all pairs of conditional output
distributions in a Gram matrix of square-root densities. Partitioning
its rows and columns by the two values of $f$ gives an off-diagonal block.
Its nuclear norm---the sum of its singular values---provides, after
normalization, a lower bound on the affinity of the two output mixtures.
Section~\ref{sec:gram} proves this reduction and a matrix lemma for
exchanging two invariant subspaces.

On the cube, the Gram matrix preserves the subspaces of even and odd
functions under input reversal. Multiplication by an antipodal $f$
preserves Euclidean norms and exchanges these subspaces.
An eigenvalue comparison gives a lower bound by matching the two spectra:
minimize the sum of absolute differences between paired eigenvalues and
divide by the input alphabet size. For equal coordinate
affinities, pairing Fourier basis functions whose index sets differ by
one fixed coordinate attains this minimum. Multiplication by that
coordinate function gives exactly this pairing, explaining sharpness.

For unequal affinities, averaging the Gram matrix with a copy obtained by
reversing selected row and column signs decreases one affinity without
increasing the matrix bound. Reducing all affinities to their minimum then
returns us to the equal-affinity calculation. Section~\ref{sec:mainproof}
carries out the spectral matching and the monotonicity argument.

\subsection{Related work}

Anantharam, Bogdanov, Chakrabarti, Jayram, and Nair
formulated the Hellinger conjecture~\cite[Conjecture~2]{ABCKN17} as
\begin{equation}\label{eq:conjecture}
 \sqrt{1-(\E f)^2}-\cZ_\rho(f)\le1-\sqrt{1-\rho^2}.
\end{equation}
For balanced $f$, this is precisely~\eqref{eq:main}; for general
$f$, the bias term cannot be omitted. The Courtade--Kumar
conjecture~\cite{CK14} concerns the analogous maximization of
$I(f(X);Y)$, and follows from this Hellinger inequality by the entropy
comparison in~\cite[Proposition~1]{ABCKN17}.
The Courtade--Kumar conjecture has recently been proved for all Boolean
functions and all noise levels by Ky and Tran~\cite[Theorem~1.2]{KT26}
and, independently, by Chen et al.~\cite[Theorem~1.1]{CGJLMNW26}.

Nonlinear noise stability studies expectations of nonlinear functions
of posterior means. Li and M\'edard~\cite{LM21} relate its extremizers
to non-interactive correlation distillation and mutual information,
while Yu~\cite{Yu23} develops a Fourier optimization approach.
Yu's later work~\cite[Theorem~4.5]{Yu26} proves that dictators are
locally optimal among balanced functions, for each $0<\rho<1$, when
the objective is twice continuously differentiable and strictly convex.
With computer-assisted estimates, it also establishes the balanced
Courtade--Kumar inequality for $\rho\in[0,0.914]$~\cite[Theorem~1.5]{Yu26}.

Chen and Nair~\cite{CN24} study the relation between dictator
optimality and isoperimetric inequalities. Chen, Gohari, and
Nair~\cite{CGN26} develop a differential-equation approach to the
Hellinger and Courtade--Kumar problems. These approaches connect
noise-dependent inequalities to geometric or differential estimates;
our argument instead uses the two spectra forced by a sign-reversing
symmetry.

Durcik, Fraccaroli, and Roos~\cite[Theorem~1.1]{DFR26} prove the weak
Hellinger conjecture for arbitrary Boolean input functions and one-bit
statistics of the noisy output. Their proof reduces to a three-parameter
inequality, established by polynomial approximations and exact
computer-assisted positivity checks. The distinction here is the
observation: we retain the full noisy output, while imposing a
sign-reversing symmetry on the input function.

The low-noise limit has also seen progress. Durcik, Ivanisvili, Roos, and
Xie~\cite[Theorem~1.1 and~(1.4)]{DIRX26} prove the sharp critical-exponent isoperimetric
inequality on the cube. In particular, if $s_f(x)$ counts the single-bit
flips that change $f(x)$, they prove $\E\sqrt{s_f(X)}\ge1$ for every
balanced Boolean function. The quantity $s_f(x)$ is called pointwise
sensitivity. This is the sharp low-noise consequence predicted by
the balanced Hellinger inequality. Appendix~\ref{app:boundary} obtains
a weighted matrix version for functions with a sign-reversing coordinate
flip, using the same Gram bound as at nonzero noise.

An announcement in the OpenAI manuscript collection~\cite{OpenAI26},
entry~119, states a proof of the Hellinger inequality for arbitrary
Boolean output bias and every noise correlation. The present paper
gives a self-contained spectral argument for sign-reversing functions
and extends it to arbitrary independent binary-input channels.

The matrix tools themselves are classical: singular-value
comparison~\cite{Mirsky60,Audenaert14} and ordered matching on the
line~\cite{Vallender74}. We prove the forms needed here. The steps
specific to the channel problem are the Gram reduction, the exact
matching of the even and odd cube spectra, and coordinatewise monotonicity.

Section~\ref{sec:extensions} describes the additional sign-reversal symmetries
and gives examples. Section~\ref{sec:applications} treats information
and testing, and Section~\ref{sec:limitations} explains what the method
does not yet cover. Appendices~\ref{app:boundary} and~\ref{app:lift}
record the edge-boundary limit and the limitation of adding one auxiliary
input bit. Appendix~\ref{app:calculations} collects two auxiliary calculations.

\section{A matrix lower bound}\label{sec:gram}

We first separate the channel geometry from the structure of the
Boolean cube. The first result applies to any finite uniform input
alphabet; the second is a matrix statement about two spaces exchanged
by a symmetry.

All matrices are real. We write $M^{\mathsf T}$ for the transpose,
$\Tr M$ for the trace, and $\norm{\cdot}_2$ for Euclidean length.
If $s_j(M)$ are the singular values of $M$, its \emph{nuclear norm}
(or trace norm) is $\norm M_1=\sum_j s_j(M)$. For symmetric $M$ it
is the sum of the absolute eigenvalues; for positive semidefinite $M$
it equals $\Tr M$. We use the triangle inequality, invariance under
orthogonal changes of coordinates, and the rank-one identity
$\norm{ab^{\mathsf T}}_1=\norm a_2\norm b_2$.

\subsection{From mixture affinity to a Gram matrix}

In this subsection, let $X$ be uniform on any finite alphabet
$\mathcal X$ of size $N$, let $f:\mathcal X\to\Q$, and let $Y$ be
the output of a channel $W$. The functional remains
$\cZ_W(f)=\E\psi(\E[f(X)\mid Y])$. Choose
a finite measure $\nu$ dominating its rows, for example
$\nu=\sum_xW(\cdot\mid x)$, and write $w_x$ for the density of row $x$.
Let $A=f^{-1}(1)$ and $B=f^{-1}(-1)$, and define
$S_A(y)=\sum_{x\in A}w_x(y)$ and $S_B(y)=\sum_{x\in B}w_x(y)$.
The output density is $(S_A+S_B)/N$, and the posterior mean is
$(S_A-S_B)/(S_A+S_B)$ wherever this density is positive. Hence
\[
 \cZ_W(f)=\frac2N\int\sqrt{S_A(y)S_B(y)}\,d\nu(y).
\]
Zero-density outputs contribute zero. If $p=\Pr(f(X)=1)\in(0,1)$
and $P_\pm=P_{Y\mid f(X)=\pm1}$, this also gives
$\cZ_W(f)=2\sqrt{p(1-p)}\,Z(P_+,P_-)$. Thus for balanced $f$ the
functional is exactly the class-conditional affinity.

Define $a_y=(\sqrt{w_x(y)})_{x\in A}$ and
$b_y=(\sqrt{w_x(y)})_{x\in B}$. The Gram matrix of the square-root
densities has entries $G(x,x')=\int\sqrt{w_xw_{x'}}\,d\nu$.
It is positive semidefinite with unit diagonal. Its off-diagonal block
$G[A,B]$, the submatrix with rows in $A$ and columns in $B$, equals
$\int a_yb_y^{\mathsf T}\,d\nu$. The integrability needed here follows
from $\norm{a_y}_2\norm{b_y}_2\le(S_A(y)+S_B(y))/2$. The triangle
inequality now gives
\[
 \norm{G[A,B]}_1
 \le\int\norm{a_yb_y^{\mathsf T}}_1\,d\nu
 =\int\sqrt{S_A S_B}\,d\nu
 =\frac N2\cZ_W(f).
\]

Let $D_f=\diag(f(x):x\in\mathcal X)$, the diagonal matrix representing
multiplication by $f$. The entry of $D_fGD_f$ at $(x,x')$ is
$f(x)f(x')G(x,x')$. Thus $G-D_fGD_f$ removes within-class entries
and doubles entries between the two classes.

\begin{proposition}[Gram matrix bound]\label{prop:gram}
For the finite uniform input model above,
\begin{equation}\label{eq:gram-bound}
 \cZ_W(f)\ge\frac1{2N}\norm{G-D_fGD_f}_1.
\end{equation}
\end{proposition}

\begin{proof}
Order the inputs with $A$ first and $B$ second, and put $C=G[A,B]$.
Then
\[
 G-D_fGD_f=\begin{pmatrix}0&2C\\2C^{\mathsf T}&0\end{pmatrix}.
\]
Its nonzero eigenvalues are $\pm2s_j(C)$, so its nuclear norm is
$4\norm C_1$. The preceding bound proves~\eqref{eq:gram-bound}.
If $f$ is constant, both sides are zero.
\end{proof}

\subsection{Exchanging two invariant subspaces}\label{sec:exchange}

A subspace is invariant under a matrix if the matrix maps it into
itself. We will use a symmetric orthogonal map $D$ that exchanges two
such spaces. Orthogonality preserves lengths, and symmetry then gives
$D^2=I$. All eigenvalue lists include multiplicities.

The following comparison follows from the identity-function case
of~\cite[Theorem~1]{Audenaert14}: add the same sufficiently large
multiple of $I$ to both matrices, making them positive semidefinite
without changing either difference. Their singular values are then
their eigenvalues. We include a direct proof.
\begin{lemma}[Nuclear-norm eigenvalue inequality]\label{lem:mirsky}
If $A,B$ are real symmetric matrices of the same size and their
eigenvalues $a_j,b_j$ are in nonincreasing order, then
$\sum_j|a_j-b_j|\le\norm{A-B}_1$.
\end{lemma}

\begin{proof}
Let $H=A-B$. Diagonalizing $H$, separate its positive and negative
eigenvalues to write $H=H_+-H_-$ with $H_+,H_-$ positive semidefinite
and $\Tr H_++\Tr H_-=\norm H_1$. Put $C=B+H_+=A+H_-$, and let
$c_j$ be its ordered eigenvalues. Adding a positive semidefinite matrix
cannot decrease any ordered eigenvalue, by the variational
characterization of eigenvalues. Thus $c_j\ge a_j,b_j$, and
\[
 \sum_j|a_j-b_j|\le\sum_j(2c_j-a_j-b_j)
 =2\Tr C-\Tr A-\Tr B
 =\norm H_1.\qedhere
\]
\end{proof}

Let $G$ be a real symmetric matrix preserving an orthogonal splitting
$\mathcal H_+\oplus\mathcal H_-$ into equal-dimensional spaces. Let
$G_+$ and $G_-$ be its restrictions, with ordered eigenvalues $a_j$
and $b_j$. These objects are fixed before the exchange map is chosen.

\begin{lemma}[Exchange of invariant subspaces]\label{prop:spectral}
If a symmetric orthogonal map $D$ exchanges $\mathcal H_+$ and
$\mathcal H_-$, then
\[
 \frac12\norm{G-DGD}_1\ge\sum_j|a_j-b_j|.
\]
The right-hand side is the minimum over all such exchange maps.
\end{lemma}

\begin{proof}
In this splitting, for an orthogonal map $U:\mathcal H_-\to\mathcal H_+$,
\[
 G=\begin{pmatrix}G_+&0\\0&G_-\end{pmatrix},\qquad
 D=\begin{pmatrix}0&U\\U^{\mathsf T}&0\end{pmatrix}.
\]
Put $K=G_+-UG_-U^{\mathsf T}$. Then $G-DGD$ has diagonal blocks
$K$ and $-U^{\mathsf T}KU$. Their nuclear norms agree, so
$\tfrac12\norm{G-DGD}_1=\norm K_1\ge\sum_j|a_j-b_j|$ by
Lemma~\ref{lem:mirsky}. To attain equality, choose $U$ to map ordered
eigenvectors of $G_-$ to ordered eigenvectors of $G_+$.
\end{proof}

Allowing every such orthogonal map lets us choose how the eigenbases
on the two invariant subspaces are aligned. The minimum therefore depends
only on their eigenvalue lists; neither positivity nor a product
structure is needed.

\section{Proof of the main theorems}\label{sec:mainproof}

We return to the cube and put $N=2^n$. For a product channel, integration
of the square-root densities factors across coordinates, giving
$G(x,x')=\prod_{i:x_i\ne x'_i}Z_i$. For arbitrary coordinate affinities
$\boldsymbol z=(z_1,\ldots,z_n)\in[0,1]^n$, use $\otimes$ for the
Kronecker product of matrices and write
\[
 G_{\boldsymbol z}=\bigotimes_{i=1}^n
 \begin{pmatrix}1&z_i\\z_i&1\end{pmatrix},\qquad
 q_f(\boldsymbol z)=\frac{\norm{G_{\boldsymbol z}-D_fG_{\boldsymbol z}D_f}_1}{2N}.
\]
Thus $\Tr G_{\boldsymbol z}=N$ and
$\cZ_W(f)\ge q_f(Z_1,\ldots,Z_n)$ by Proposition~\ref{prop:gram}.
When every coordinate affinity equals $h$,
write $G_h$ and $q_f(h)$. In a binary symmetric channel
$h=\sqrt{1-\rho^2}$; the coordinate affinity $h$ and the noise correlation
$\rho$ play different roles.

\subsection{Even and odd Fourier subspaces}\label{sec:anticommute}

For $T\subseteq[n]$, the Walsh (parity) function is
$\chi_T(x)=\prod_{i\in T}x_i$, with $\chi_\varnothing=1$.
The vectors $N^{-1/2}(\chi_T(x))_x$ form an orthonormal basis of
$\R^N$. The degree of this basis vector is $|T|$. Let
$(Ju)(x)=u(-x)$ be the global flip, and let $\mathcal H_+$ and
$\mathcal H_-$ be its even and odd subspaces: $Ju=u$ and $Ju=-u$,
respectively. Because $\chi_T(-x)=(-1)^{|T|}\chi_T(x)$, even degrees
span $\mathcal H_+$ and odd degrees span $\mathcal H_-$. Each has
dimension $N/2$.

The Gram entries are unchanged when both indices are reversed, so
$G_{\boldsymbol z}J=JG_{\boldsymbol z}$ and both spaces are invariant.
Antipodality gives the opposite relation,
\begin{equation}\label{eq:anticommutation}
 D_fJ=-JD_f.
\end{equation}
This relation is called anticommutation: multiplication by $f$ sends
even functions to odd functions and conversely. Since $f^2=1$,
the matrix $D_f$ is symmetric and orthogonal, as required by
Lemma~\ref{prop:spectral}. Hence
$Nq_f(\boldsymbol z)\ge\sum_j|a_j-b_j|$, where $a_j,b_j$ are the
ordered eigenvalues on the two subspaces. The right-hand side no longer
depends on $f$.

\subsection{The exact matching cost}\label{sec:cost}

For equal coordinate affinities $h$, each $2\times2$ factor has
eigenvectors $(1,1)^{\mathsf T}$ and $(1,-1)^{\mathsf T}$, with eigenvalues
$1+h$ and $1-h$. Their Kronecker products show that each Walsh function
of degree $d$ is an eigenvector of $G_h$ with eigenvalue
\[
 \lambda_d=(1+h)^{n-d}(1-h)^d,
 \qquad\text{with multiplicity }\binom nd.
\]
Thus the even list $a$ and odd list $b$ contain alternate degree levels.
To compare these lists, it is useful to count how many eigenvalues lie
above each threshold. The following identity for finite ordered lists~\cite{Vallender74}
expresses the matching cost in terms of these counts and also applies
to other pairs of spectra.

\begin{lemma}[Matching by threshold counts]\label{lem:counts}
For two nonnegative lists of equal length, write $a_j,b_j$ in
nonincreasing order. The minimum total absolute difference over all
pairings is
\[
 \sum_j|a_j-b_j|
 =\int_0^\infty
 \left|\#\{j:a_j>t\}-\#\{j:b_j>t\}\right|\,dt.
\]
\end{lemma}

\begin{proof}
At a threshold $t$, the excess number of points in either list above
$t$ must be paired across that threshold. Ordered pairing has exactly
this many crossings, since its above-threshold indices are initial
segments. Integrate over $t$: the length of each paired interval is
the absolute difference of its endpoints. This proves both the lower
bound for any pairing and equality for the ordered one.
\end{proof}

\begin{lemma}[Parity matching]\label{lem:cost}
The even and odd eigenvalue lists of $G_h$ satisfy
$\sum_j|a_j-b_j|=Nh$ for $0\le h\le1$.
\end{lemma}

\begin{proof}
We first identify an optimal pairing, then compute its cost.
Suppose $0<h<1$, so the levels strictly decrease. Fix a coordinate
$i$ and pair $R$ with $R\cup\{i\}$ for every
$R\subseteq[n]\setminus\{i\}$. Each pair joins one even degree to
one odd degree, and the two degrees are adjacent.

Consider a threshold between $\lambda_{d+1}$ and $\lambda_d$.
Only pairs joining degrees $d$ and $d+1$ cross it. All their upper
endpoints therefore belong to the same parity list. Pairs wholly
above the threshold contribute equally to the two counts, so the
number of crossing pairs is exactly the absolute difference of those
counts. The pairing attains the threshold lower bound in
Lemma~\ref{lem:counts} at every gap, and is optimal.

There are $\binom{n-1}{d}$ pairs joining degrees $d$ and $d+1$.
Their total cost is
\[
 \sum_j|a_j-b_j|
 =\sum_{d=0}^{n-1}\binom{n-1}{d}
    2h(1+h)^{n-1-d}(1-h)^d
 =2h\,2^{n-1}=Nh.
\]
At $h=0$, both lists consist of ones. At $h=1$, the even list has
one eigenvalue $N$ and all remaining eigenvalues in both lists are
zero. These give the same formula at the endpoints.
\end{proof}

Multiplication by $x_i$ realizes this optimal pairing: it sends
$\chi_T$ to the Walsh function obtained by adding or removing $i$
from $T$. Thus the exchange used to compute the lower bound is
itself induced by a Boolean function, namely a dictator.

\begin{theorem}[Equal coordinate affinities]\label{thm:nuclear}
Every antipodal Boolean function satisfies $q_f(h)\ge h$ for
$0\le h\le1$, with equality for signed dictators.
\end{theorem}

\begin{proof}
Apply Lemmas~\ref{prop:spectral} and~\ref{lem:cost}. For a
dictator, the preceding coordinate pairing is its actual action on
the Walsh basis. Thus
$\tfrac12\norm{G_h-D_fG_hD_f}_1=Nh$ and $q_f(h)=h$.
Changing the sign of the dictator leaves $D_fG_hD_f$ unchanged.
\end{proof}

Taking $h=\sqrt{1-\rho^2}$ in the Gram matrix bound proves
Theorem~\ref{thm:main}. For its equality case, the posterior mean of
a signed dictator is $\pm\rho Y_i$, so its affinity is exactly
$\sqrt{1-\rho^2}$.

\subsection{Unequal coordinate channels}

For equal coordinate affinities, the eigenvalues are ordered by degree.
For unequal affinities that ordering may fail. We therefore decrease each
affinity to the smallest one. The following monotonicity property
justifies this step and holds for every Boolean function.

\begin{lemma}[Coordinatewise monotonicity]\label{lem:monotonicity}
For every Boolean function $f$, the map $q_f(\boldsymbol z)$ is
nondecreasing in each coordinate on $[0,1]^n$.
\end{lemma}

\begin{proof}
Let $R_i=\diag(x_i:x\in\Q^n)$, and obtain $\boldsymbol z'$ from
$\boldsymbol z$ by replacing $z_i$ with $cz_i$, where $0\le c\le1$.
The coordinate factorization gives
\[
 G_{\boldsymbol z'}=\frac{1+c}{2}G_{\boldsymbol z}
                  +\frac{1-c}{2}R_iG_{\boldsymbol z}R_i.
\]
Since $R_iD_f=D_fR_i$, the same convex combination holds for
$\Delta_f(\boldsymbol z)=G_{\boldsymbol z}-D_fG_{\boldsymbol z}D_f$.
Convexity and invariance under orthogonal changes of coordinates imply
$\norm{\Delta_f(\boldsymbol z')}_1\le\norm{\Delta_f(\boldsymbol z)}_1$.
Any decrease of a positive coordinate has this form; a zero coordinate
needs no change.
\end{proof}

\begin{proof}[Proof of Theorem~\ref{thm:channels}]
Set $h=\min_iZ_i$. Proposition~\ref{prop:gram},
Lemma~\ref{lem:monotonicity}, and Theorem~\ref{thm:nuclear} give
$\cZ_W(f)\ge q_f(Z_1,\ldots,Z_n)\ge q_f(h)\ge h$.
For $f(x)=x_j$, the conditional output laws differ only in coordinate
$j$, and affinity is multiplicative under products; hence $\cZ_W(f)=Z_j$.
\end{proof}

This monotonicity concerns the matrix lower bound; it does not require
the coordinate channels to be ordered by stochastic degradation.
The distinction matters because directly matching the spectra for unequal
coordinate affinities can lose the sharp constant.
For example, $\boldsymbol z=(1/4,1/2,1/2,3/4)$ has normalized
matching cost $15/64<\min_i z_i=1/4$; the exact calculation is in
Appendix~\ref{app:heterogeneous}.
Decreasing the affinities while keeping $D_f$ fixed preserves the
restriction that $D_f$ is multiplication by a Boolean function.
Unrestricted eigenvalue matching can lose this restriction.

\section{Other sign-reversing symmetries}\label{sec:extensions}

We retain uniform input bits and the product channel
$W=\bigotimes_i W_i$. The proof also applies when reversing a fixed
subset of coordinates changes the sign of $f$. Concavity then gives
the corresponding bound for a randomized binary function of $X$.

\subsection{A fixed subset of reversed coordinates}

For $S\subseteq[n]$, write $x^{\oplus S}$ for the input obtained by
reversing coordinates in $S$. If this flip reverses $f$, freezing the
coordinates outside $S$ leaves an antipodal function on each slice. A nonempty flip has
no fixed points, so it also pairs the two classes.

\begin{theorem}[Sign reversal under coordinate flips]\label{thm:fixedflip}
If a nonempty set $S$ satisfies $f(x^{\oplus S})=-f(x)$ for every
$x$, then $f$ is balanced and $\cZ_W(f)\ge\min_{i\in S}Z_i$.
The stronger matrix bound $q_f(\boldsymbol z)\ge\min_{i\in S}z_i$
also holds for every $\boldsymbol z\in[0,1]^n$.
\end{theorem}

\begin{proof}
Put $R=[n]\setminus S$ and $f_u(v)=f(v,u)$ for $u\in\Q^R$.
Each slice is antipodal on $\Q^S$. Decrease the coordinates $z_i$ for
$i\in R$ to zero, leaving those in $S$ unchanged, and denote the result
by $\boldsymbol z'$. The zero-affinity factors are identity matrices, so
ordering inputs by $u$ makes the Gram difference block diagonal. Hence
\[
 q_f(\boldsymbol z)\ge q_f(\boldsymbol z')
 =2^{-|R|}\sum_{u\in\Q^R}q_{f_u}(\boldsymbol z_S)
 \ge\min_{i\in S}z_i.
\]
The last inequality applies Lemma~\ref{lem:monotonicity} and
Theorem~\ref{thm:nuclear} to each slice. The Gram matrix bound proves
the channel assertion. Pairing $x$ with $x^{\oplus S}$ proves
balancedness.
\end{proof}

In channel terms, setting the affinities outside $S$ to zero corresponds
to revealing those input bits. Conditional Jensen says this can only
decrease the affinity. The block calculation proves the stronger
matrix version of the same observation.
The same symmetry has a Fourier description. The Fourier--Walsh expansion is
$f(x)=\sum_{T\subseteq[n]}\widehat f(T)\chi_T(x)$, where
$\widehat f(T)=\E[f(X)\chi_T(X)]$ is its Fourier coefficient; see
\cite{ODonnell14}. Reversing coordinates in $S$ multiplies $\chi_T$
by $(-1)^{|T\cap S|}$, giving
\begin{equation}\label{eq:fourier-flip}
 f(x^{\oplus S})=-f(x)
 \quad\Longleftrightarrow\quad
 \widehat f(T)=0\ \text{whenever }|T\cap S|\text{ is even}.
\end{equation}
If several sign-reversing sets exist, one may maximize the lower bound
$\min_{i\in S}Z_i$ over them. For example,
$f(x)=x_i u(x_{[n]\setminus\{i\}})$ has the bound $Z_i$ for every Boolean
function $u$ of the remaining coordinates, regardless of global
antipodality. A nonempty parity function $\chi_T$ has
bound $\max_{i\in T}Z_i$; for homogeneous binary symmetric noise its exact
value is $\sqrt{1-\rho^{2|T|}}$.

\subsection{Randomization}

Let $g:\Q^n\to[-1,1]$, and generate a binary random variable $B\in\Q$
with $\Pr(B=1\mid X=x)=(1+g(x))/2$. Assume that $B$ and $Y$ are
conditionally independent given $X$. When $B$ is balanced, meaning
$\Pr(B=1)=1/2$, the posterior identity from Section~\ref{sec:gram} gives
$Z(P_{Y\mid B=1},P_{Y\mid B=-1})=\E\psi(\E[g(X)\mid Y])$.
Under the sign-reversal condition below, $g$ is a convex combination of
Boolean functions with the same symmetry. Concavity of $\psi$ then
transfers the deterministic bound to $B$.

\begin{proposition}\label{prop:randomized}
If $g(x^{\oplus S})=-g(x)$ for a nonempty $S$, then $B$ is balanced
and $Z(P_{Y\mid B=1},P_{Y\mid B=-1})\ge\min_{i\in S}Z_i$ under
independent coordinate channels. A signed dictator on a coordinate of
minimum affinity in $S$ attains equality.
\end{proposition}

\begin{proof}
Choose one representative from each pair $\{x,x^{\oplus S}\}$.
Independently of $(X,Y)$, assign signs $F(x)$ to these representatives
with means $g(x)$, and put $F(x^{\oplus S})=-F(x)$.
Every realization satisfies the same sign-reversal condition, with
$\E_F[F(x)]=g(x)$ for every $x$.
Concavity yields
\[
 \E_Y\psi(\E[g(X)\mid Y])
 \ge\E_F\E_Y\psi(\E[F(X)\mid Y,F])\ge\min_{i\in S}Z_i.
\]
Pairing inputs shows $\E g(X)=0$, so $B$ is balanced. The posterior
identity gives the stated affinity bound. A signed dictator on a coordinate
of minimum affinity in $S$ proves sharpness.
\end{proof}

Concavity is needed here because multiplication by $g$ need not
preserve lengths. The diagonal matrix with entries $g(x)$ need not
square to the identity, so Lemma~\ref{prop:spectral} is applied to
the deterministic functions in the mixture.

\subsection{Examples}

A linear threshold function with threshold zero has the form
$f(x)=\sgn(\sum_i a_ix_i)$ with real weights.
Here $\sgn(t)=1$ for $t>0$ and $-1$ for $t<0$. If the weighted sum
never vanishes, the function is antipodal; majority on an odd number of
bits is an example. When ties occur, opposite assignments on paired
vertices preserve antipodality. Fair independent tie-breaking instead
has conditional mean $g(x)=\sgn(\sum_i a_ix_i)$ with $\sgn(0)=0$,
so Proposition~\ref{prop:randomized} applies. The same proposition
covers any randomized binary function with conditional mean
$g(x)=\tanh(a^{\mathsf T}x)$, for a real weight vector $a$.

Composition preserves antipodality. Consider a Boolean circuit with
inputs $\pm x_i$, no constant inputs, and gates that reverse their
output when all their inputs are reversed. The sign-reversal property
propagates from inputs to outputs, so the computed function is antipodal. Variables and gates may be shared, provided the circuit
has no directed cycles. The independent variables are the original
input bits; the intermediate gate outputs need not be independent.
\section{Information and testing}\label{sec:applications}

We now translate the affinity bound into bounds on mutual information
and hypothesis-testing errors. The Gaussian example shows how stochastic
degradation from a product channel also handles a class of channels
with correlated outputs.

\subsection{Divergence, error probability, and mutual information}

Let $B$ be a binary random variable generated from $X$ as in
Theorem~\ref{thm:fixedflip} or Proposition~\ref{prop:randomized}.
Thus $g(x)=\E[B\mid X=x]$ changes sign when the coordinates in a fixed
nonempty set $S$ are reversed; in the deterministic case, $B=f(X)$ and $g=f$.
Put $P_\pm=P_{Y\mid B=\pm1}$ and $z_* =\min_{i\in S}Z_i$.
The identity $D_{1/2}(P\Vert Q)=-2\log Z(P,Q)$ for R\'enyi
divergence~\cite{vEH14} immediately gives
$D_{1/2}(P_+\Vert P_-)\le-2\log z_*$ in nats, with
$-\log0=+\infty$.
For binary symmetric noise the bound is $-\log(1-\rho^2)$, sharp for a
dictator.

For an integer $k\ge1$, let $P_e^{(k)}$ be the minimum error probability in
equal-prior testing between $P_+^{\otimes k}$ and $P_-^{\otimes k}$,
the laws of $k$ independent observations from one of the two classes.

\begin{corollary}[Bayes error]\label{cor:bayes}
The minimum error probability satisfies
\[
 P_e^{(k)}\ge\frac{1-\sqrt{1-z_*^{\,2k}}}{2}.
\]
For $k=1$ and binary symmetric channels, a signed dictator on a coordinate
of minimum affinity within $S$ attains the bound.
\end{corollary}

\begin{proof}
For densities $p,q$ of laws $P,Q$, their total variation distance is
$\TV(P,Q)=\tfrac12\int|p-q|\,d\nu$. Apply Cauchy--Schwarz to
$|p-q|=|\sqrt p-\sqrt q|(\sqrt p+\sqrt q)$ to obtain
$\TV(P,Q)\le\sqrt{1-Z(P,Q)^2}$. Use
$P_e=(1-\TV)/2$ and the product identity
$Z(P_+^{\otimes k},P_-^{\otimes k})=Z(P_+,P_-)^k\ge z_*^k$.
For a coordinate $j\in S$ attaining $z_*$, the choice $B=\pm X_j$
under binary symmetric noise gives
$\TV(P_+,P_-)=\sqrt{1-Z_j^2}$, proving equality.
\end{proof}

An upper bound on the Chernoff information requires more than this
affinity estimate. Define
$C(P,Q)=-\log\inf_{0<s<1}\int p^sq^{1-s}\,d\nu$, where $p,q$
are densities relative to a common dominating measure. The open interval
avoids endpoint conventions when supports differ.
In general $C(P,Q)\ge-\log Z(P,Q)$. The reverse estimate needed here
comes from an output symmetry.

The useful condition is that a measurable map $\tau$, with
$\tau(\tau(y))=y$, exchanges the two output laws. We state the resulting
fact separately because the Gaussian application uses it without a
product channel.

\begin{lemma}[Symmetric testing laws]\label{lem:symmetric-testing}
If such a map exchanges $P$ and $Q$, then $C(P,Q)=-\log Z(P,Q)$.
\end{lemma}

\begin{proof}
The measure $\nu=P+Q$ is preserved by $\tau$. The densities satisfy
$q(y)=p(\tau(y))$ almost everywhere with respect to $\nu$, so $F(s)=\int p^sq^{1-s}\,d\nu$ obeys
$F(s)=F(1-s)$. Cauchy--Schwarz gives
$F(1/2)\le\sqrt{F(s)F(1-s)}=F(s)$. Thus the minimum is the affinity.
\end{proof}

Call channel $i$ \emph{output-symmetric} if a measurable map $\tau_i$
with $\tau_i^2=\mathrm{id}$ exchanges its two rows.

\begin{corollary}[Chernoff information]\label{cor:chernoff}
If the channels on $S$ are output-symmetric, then
$C(P_+,P_-)=-\log Z(P_+,P_-)\le-\log z_*$.
A signed dictator on a coordinate of minimum affinity in $S$ attains the bound.
\end{corollary}

\begin{proof}
Apply $\tau_i$ on output coordinates in $S$ and the identity elsewhere.
The resulting map exchanges $P_+$ and $P_-$ because
$g(x^{\oplus S})=-g(x)$. Lemma~\ref{lem:symmetric-testing} and
$Z(P_+,P_-)\ge z_*$ prove the bound. A dictator reduces both quantities
to those of its chosen coordinate channel.
\end{proof}

Under independent repetitions from the two class-conditional laws,
Chernoff's theorem~\cite{Chernoff52} identifies $C(P_+,P_-)$ as the
optimal equal-prior error exponent. Here each repetition uses a fresh
input drawn from the same class (or from $P_{X\mid B=b}$ for a randomized
binary variable $B$). Repeatedly observing one fixed hidden input is a
different model.

The entropy comparison of~\cite[Proposition~1]{ABCKN17} also applies
to balanced binary random variables generated from $X$. Write
$h_2(p)=-p\log_2p-(1-p)\log_2(1-p)$ for binary entropy, with
$0\log_2 0=0$. As throughout the paper, entropy and mutual information
are measured in bits.

\begin{corollary}[Mutual information]\label{cor:shannon}
For every binary variable $B$ above, under binary symmetric noise with
common correlation $\rho$,
\[
 I(B;Y)\le1-h_2\!\left(\frac{1-\rho}{2}\right),\qquad 0\le\rho\le1.
\]
\end{corollary}

\begin{proof}
The scalar function $F(z)=h_2((1-\sqrt{1-z^2})/2)$ is increasing
and convex on $[0,1]$~\cite[Lemma~1]{ABCKN17};
Appendix~\ref{app:entropy} records the calculus verification.
For the posterior mean $m=\E[B\mid Y]$, symmetry of binary entropy gives
$F(\psi(m))=h_2((1-m)/2)$. Thus the conditional entropy can be bounded
by applying Jensen's inequality to $F$ and then using the affinity bound:
\[
 H(B\mid Y)=\E F(\psi(m))\ge F(\E\psi(m))
 \ge F(\sqrt{1-\rho^2})=h_2((1-\rho)/2).
\]
Finally, $H(B)=1$.
\end{proof}

\subsection{Independent and correlated Gaussian observations}

Consider first $Y_i=a_iX_i+\sigma_iN_i$, where the $N_i$ are independent
standard normals, independent of $X$, and $\sigma_i>0$.
The affinity of $\mathcal N(a_i,\sigma_i^2)$ and
$\mathcal N(-a_i,\sigma_i^2)$ is
$Z_i=e^{-a_i^2/(2\sigma_i^2)}$. Reflection $y_i\mapsto-y_i$
exchanges the two rows, so these channels are output-symmetric.
Whenever $g(x)=\E[B\mid X=x]$ is antipodal, as in Section~\ref{sec:extensions},
Corollary~\ref{cor:chernoff} gives
\begin{equation}\label{eq:diagonal-gaussian}
 Z(P_+,P_-)\ge e^{-\frac12\max_i a_i^2/\sigma_i^2},
 \qquad C(P_+,P_-)\le\frac12\max_i\frac{a_i^2}{\sigma_i^2}.
\end{equation}
Both bounds are sharp for each fixed diagonal channel, attained by a
dictator on a coordinate of largest signal-to-noise ratio. For $B=f(X)$,
the conditional laws are mixtures of $2^{n-1}$ Gaussians, yet their
optimal testing exponent cannot exceed that of this single coordinate.
If $g(x^{\oplus S})=-g(x)$, the maximum
in~\eqref{eq:diagonal-gaussian} can be restricted to $i\in S$.

For correlated Gaussian observations, the Gram matrix need not have a
Kronecker product structure. We instead realize the observation as a
stochastically degraded output of a product Gaussian channel: its output
is transformed linearly and independent noise is added.
The affinity bound then follows from data processing.
If $B\to U\to V$ is a Markov chain, conditional Jensen gives
\begin{equation}\label{eq:processing}
 \E\psi(\E[B\mid V])\ge\E\psi(\E[B\mid U]).
\end{equation}
This is Hellinger data processing, expressed through posterior means;
see~\cite[Theorem~9 and Example~2]{vEH14}. It is consistent with the ordering of statistical
experiments by randomization~\cite{Blackwell53}.

Now let $Y=AX+\eta$, where $A\in\R^{d\times n}$ and
$\eta\sim\mathcal N(0,\Sigma)$ is independent of $X$, with positive
definite covariance $\Sigma$. Let $B$ be a binary random variable whose conditional
mean $g(x)=\E[B\mid X=x]$ is antipodal, and assume that $B$ and
$\eta$ are conditionally independent given $X$. Retain
$P_\pm=P_{Y\mid B=\pm1}$. Put
$L=\norm{\Sigma^{-1/2}A}_{\mathrm{op}}^2$, where
$\norm M_{\mathrm{op}}=\sup_{\norm v_2=1}\norm{Mv}_2$ is the largest
singular value of $M$. Multiplication by $\Sigma^{-1/2}$ makes the noise
covariance the identity; this is the usual whitening transformation.

\begin{proposition}[Linear Gaussian observations]\label{prop:linear-gaussian}
In this model, $Z(P_+,P_-)\ge e^{-L/2}$ and $C(P_+,P_-)\le L/2$.
Both constants are sharp among models with the same value of $L$.
\end{proposition}

\begin{proof}
Whitening is invertible, so it preserves affinity and Chernoff
information. Write $M=\Sigma^{-1/2}A$. If $L=0$, the output is
independent of $B$. For $L>0$, consider the product Gaussian channel
$U=\sqrt L\,X+\xi$, where $\xi\sim\mathcal N(0,I_n)$ is independent
of $(X,B)$. Independently of $(X,B,\xi)$, introduce
$\zeta\sim\mathcal N(0,I_d-MM^{\mathsf T}/L)$; this covariance is
positive semidefinite because $\norm M_{\mathrm{op}}^2=L$.
A singular covariance is allowed. Then
\[
 V=MU/\sqrt L+\zeta
   =MX+M\xi/\sqrt L+\zeta,
 \qquad V\mid X\sim\mathcal N(MX,I_d).
\]
Thus $(B,V)$ has the same law as $(B,\Sigma^{-1/2}Y)$, while
$B\to U\to V$ is a Markov chain. Apply
\eqref{eq:processing} and~\eqref{eq:diagonal-gaussian} to obtain
$Z(P_+,P_-)\ge e^{-L/2}$. Antipodality and centered Gaussian noise
give $p_-(y)=p_+(-y)$ for the class-conditional densities. Lemma~\ref{lem:symmetric-testing} yields
$C(P_+,P_-)=-\log Z(P_+,P_-)$. For sharpness, take $\Sigma=I_d$,
$B=X_1$, and $A$ with sole possibly nonzero entry $A_{11}=\sqrt L$.
Only the first observation depends on $B$, and its affinity
is $e^{-L/2}$.
\end{proof}

For a given nondiagonal $A$, the bound need not be attained by a
dictator. Sharpness concerns the class of models determined only by
$L$. The input remains uniform on the whole cube, as required by the
product channel used in the proof.

\section{Scope of the argument}\label{sec:limitations}

\subsection{Balancedness does not imply sign-reversal symmetry}

Functions whose sign is reversed by flipping a fixed subset of coordinates
form a proper subclass of balanced functions. This remains true among
monotone functions, for which changing an input bit from $-1$ to $1$
never decreases the output. With sign-valued
$\mathrm{OR}(u,v)=(1+u+v-uv)/2$ and
$\mathrm{AND}(u,v)=(-1+u+v+uv)/2$, consider
\[
 f(x_1,x_2,x_3,x_4)=
 \begin{cases}
  \mathrm{OR}(x_2,x_3),&x_1=1,\\
  \mathrm{AND}(x_2,x_4),&x_1=-1.
 \end{cases}
\]
The two slice means are $1/2$ and $-1/2$, so $f$ is balanced.
Its singleton Fourier coefficients are $1/2,1/2,1/4,1/4$, while
$\widehat f(\{2,3\})=-1/4$. By~\eqref{eq:fourier-flip}, a sign-reversing
set would have to contain all four coordinates, contradicting this
nonzero even-degree coefficient. The function is monotone because each
branch is increasing and the OR branch always dominates the AND branch.

This example separates the two roles of symmetry. Balance says that the
classes have equal size. The proof uses a symmetry of the row-affinity
Gram matrix that exchanges the classes. At the matrix level, the former gives
$\Tr D_f=0$, whereas the latter gives~\eqref{eq:anticommutation}. Without an exchange of two invariant
subspaces, the same spectral comparison gives no positive constant:
$G_h$ and $D_fG_hD_f$ have identical full spectra. Extending the method
therefore requires additional control of the relative position of their
eigenspaces, rather than a comparison of their full eigenvalue lists.

Adding an extra input bit can turn any Boolean function into an
antipodal one, but does not immediately prove the general conjecture.
The added bit introduces an extra noise parameter; removing that noise
also removes the useful lower bound. Appendix~\ref{app:lift} gives the
exact comparison.

\subsection{Equality and scope}

Signed dictators prove sharpness without requiring a classification of all
extremizers. To classify all extremizers, one must track equality in the Gram
triangle inequality and the spectral comparison, and for unequal
channels in the coordinatewise monotonicity step as well. At $\rho=0$ and $\rho=1$, every balanced function has
equality in~\eqref{eq:main}; an interior classification must treat these
endpoints separately. Quantitative stability would additionally have
to control the losses in these successive steps.

The product-channel theorem assumes independent uniform input bits.
Biased inputs introduce unequal weights into the Gram construction and
generally destroy its commutation with the global flip. Arbitrary
correlated channels need not have a Gram matrix with the Kronecker
product structure used here.
Proposition~\ref{prop:linear-gaussian} covers a particular correlated
observation model through stochastic degradation from a product
channel; it does not remove those assumptions from
Theorem~\ref{thm:channels}. Beyond these cases, the general biased
conjecture~\eqref{eq:conjecture} and the balanced problem without a
sign-reversing symmetry remain outside the present argument.

\appendix

\section{A weighted edge-boundary limit}\label{app:boundary}

Two cube vertices are neighbors if they differ in exactly one
coordinate; edges joining opposite values of $f$ form its edge boundary.
When the coordinate affinities are small, Gram entries between neighbors
are first order, while entries between vertices differing in at least two
coordinates are higher order. Thus, to first order, the off-diagonal block
between the two level sets records only edges crossing this boundary.
The nuclear-norm lower bound consequently yields a weighted edge-boundary
inequality. Let
$\boldsymbol w=(w_1,\ldots,w_n)$ have nonnegative entries and let $A_w$ be the weighted cube adjacency matrix:
$A_w(x,x^{\oplus\{i\}})=w_i$, with other entries zero. For
$A=f^{-1}(1)$ and $B=f^{-1}(-1)$, its off-diagonal block $A_w[A,B]$ records the
edges cut by $f$. Write $\one_{\{\cdot\}}$ for the indicator of an event.

\begin{proposition}\label{prop:boundary}
If $f(x^{\oplus S})=-f(x)$ for a nonempty $S$, then
\[
 \frac2N\norm{A_w[A,B]}_1\ge\min_{i\in S}w_i,
 \qquad
 \E\sqrt{\sum_i w_i^2\one_{\{f(X)\ne f(X^{\oplus\{i\}})\}}}
 \ge\min_{i\in S}w_i.
\]
Both constants are attained by a signed dictator on a minimum-weight coordinate
in $S$.
\end{proposition}

\begin{proof}
For fixed $n$ and $\boldsymbol w$, as $t\downarrow0$ one has
$G_{t\boldsymbol w}=I+tA_w+O(t^2)$ in nuclear norm. Taking $t$ small
enough that $t w_i\le1$, Theorem~\ref{thm:fixedflip} and continuity give
\[
 \lim_{t\downarrow0}\frac{q_f(t\boldsymbol w)}t
 =\frac1{2N}\norm{A_w-D_fA_wD_f}_1
 =\frac2N\norm{A_w[A,B]}_1\ge\min_{i\in S}w_i.
\]
Write $r(x)=(\sum_iw_i^2\one_{\{f(x)\ne f(x^{\oplus\{i\}})\}})^{1/2}$.
Decomposing the off-diagonal block into its rows gives
$\norm{A_w[A,B]}_1\le\sum_{x\in A}r(x)$. The coordinate flip is a bijection
between $A$ and $B$ preserving $r$, so
$(2/N)\sum_{x\in A}r(x)=\E r(X)$. This proves the second bound.
For a dictator, the off-diagonal block is its coordinate weight times a
permutation matrix.
\end{proof}

When $w_i=1$, the scalar consequence is $\E\sqrt{s_f(X)}\ge1$,
already known for all balanced Boolean functions by~\cite[(1.4)]{DIRX26}.
The additional statement here is the nuclear-norm bound and its weighted
form under a sign-reversing coordinate flip. The limit is a small-noise limit because
$h=\sqrt{1-\rho^2}$.

\section{An extension by one input bit}\label{app:lift}

An arbitrary Boolean function admits the odd--even decomposition
$f_{\mathrm o}(x)=(f(x)-f(-x))/2$ and
$f_{\mathrm e}(x)=(f(x)+f(-x))/2$. At each point one component is zero
and the other is a sign. Consequently
$g(x_0,x)=f_{\mathrm o}(x)+x_0 f_{\mathrm e}(x)$ is Boolean and antipodal.
This construction extends $f$ to an antipodal function by adding one input bit.
For fixed $\rho\in[0,1]$, write
$M_\rho(s)=\E\psi(T_\rho f_{\mathrm o}+sT_\rho f_{\mathrm e})$,
where $s\in[0,1]$ controls the correlation of the added bit and all
noise-operator outputs are evaluated at the uniform noisy input.
In particular, $M_\rho(1)=\cZ_\rho(f)$ because
$f_{\mathrm o}+f_{\mathrm e}=f$.
The channel theorem gives the following bound for the extension.

\begin{proposition}\label{prop:lift}
For $0\le s\le1$,
\begin{equation}\label{eq:lift}
 M_\rho(s)\ge\min\{\sqrt{1-\rho^2},\sqrt{1-s^2}\}.
\end{equation}
For fixed $\rho$, $M_\rho$ is nonincreasing on $[0,1]$.
\end{proposition}

\begin{proof}
Use correlation $\rho$ on the original bits and correlation $s$ on the
extra bit. The posterior mean of $g$ is $a(y)+s y_0b(y)$, where
$a=T_\rho f_{\mathrm o}$ is odd and $b=T_\rho f_{\mathrm e}$ is even.
Changing $y$ to $-y$ and using the evenness of $\psi$ gives
$\E\psi(a+sb)=\E\psi(a-sb)$. Averaging over $y_0$ and applying
Theorem~\ref{thm:channels} proves~\eqref{eq:lift}.
For $|s|\le1$, the function $f_{\mathrm o}+s f_{\mathrm e}$ takes values
in $[-1,1]$, so $M_\rho$ extends to an even concave function on $[-1,1]$.
Such a function is nonincreasing on $[0,1]$.
\end{proof}

Taking $s=\rho$ gives the dictator bound with one additional factor of
$\rho$ on every even-degree Fourier coefficient. Recovering the original
functional requires $s=1$, where the bound becomes zero. Balancedness
removes the constant coefficient but not the remaining even degrees.
Thus a quantitative bound on the loss from $M_\rho(\rho)$ to
$M_\rho(1)$ would be needed to recover the balanced conjecture by this
extension. Concavity alone supplies the wrong direction for that purpose.

\section{Auxiliary calculations}\label{app:calculations}

\subsection{Convexity of the entropy comparison function}\label{app:entropy}

The function used in Corollary~\ref{cor:shannon} is
$F(z)=h_2((1-\sqrt{1-z^2})/2)$ for $0\le z\le1$.
For $0<z<1$, put $u=\sqrt{1-z^2}$. Differentiation gives
\[
 F'(z)=\frac{z}{2u\ln2}\ln\frac{1+u}{1-u},\qquad
 F''(z)=\frac{\ln((1+u)/(1-u))-2u}{2u^3\ln2}.
\]
Both derivatives are nonnegative: the logarithm is positive, and
$\ln((1+u)/(1-u))=2\int_0^u(1-t^2)^{-1}\,dt\ge2u$.
Continuity, with $F(0)=0$ and $F(1)=1$, extends monotonicity and
convexity to the closed interval. No endpoint substitution into the
derivative formulas is needed.

\subsection{Spectral matching with unequal affinities}\label{app:heterogeneous}

For unequal coordinate affinities, the Walsh eigenvalue indexed by
$T\subseteq[n]$ is
$\prod_{i\notin T}(1+z_i)\prod_{i\in T}(1-z_i)$, by the same
coordinate diagonalization as in Section~\ref{sec:cost}.
For $\boldsymbol z=(1/4,1/2,1/2,3/4)$, sorting by parity gives
\[
 \begin{aligned}
  64a&=(315,63,63,35,27,15,15,3),\\
  64b&=(189,105,105,45,21,9,9,5).
 \end{aligned}
\]
The absolute differences sum to $240/64$. Since $N=16$, the
normalized cost is $N^{-1}\sum_j|a_j-b_j|=15/64$, below
$\min_i z_i=1/4$. This exact calculation explains why the proof uses
coordinatewise monotonicity after the equal-affinity spectral comparison.


\begin{thebibliography}{99}

\bibitem{ABCKN17}
V.~Anantharam, A.~Bogdanov, A.~Chakrabarti, T.~S.~Jayram, and C.~Nair,
\emph{A conjecture regarding optimality of the dictator function under
Hellinger distance}, Information Theory and Applications Workshop
(ITA), San Diego, 2017.
\href{https://chandra.ie.cuhk.edu.hk/pub/papers/HC/hel-conj.pdf}{Author manuscript}.

\bibitem{Arikan09}
E.~Ar{\i}kan, \emph{Channel polarization: A method for constructing
capacity-achieving codes for symmetric binary-input memoryless channels},
IEEE Transactions on Information Theory \textbf{55} (2009), no.~7,
3051--3073.
\href{https://doi.org/10.1109/TIT.2009.2021379}{doi:10.1109/TIT.2009.2021379}.

\bibitem{Audenaert14}
K.~M.~R.~Audenaert, \emph{A generalisation of Mirsky's singular value
inequalities}, 2014.
\href{https://arxiv.org/abs/1410.4941v1}{arXiv:1410.4941v1}.

\bibitem{Blackwell53}
D.~Blackwell, \emph{Equivalent comparisons of experiments},
The Annals of Mathematical Statistics \textbf{24} (1953), no.~2,
265--272.
\href{https://doi.org/10.1214/aoms/1177729032}{doi:10.1214/aoms/1177729032}.

\bibitem{CGJLMNW26}
Z.~Chen, A.~Gohari, A.~Javanmard, H.~Lin, V.~Mirrokni, C.~Nair,
and D.~P.~Woodruff, \emph{A proof of the most informative Boolean
function conjecture}, 2026.
\href{https://arxiv.org/abs/2609.24931v2}{arXiv:2609.24931v2}.

\bibitem{CGN26}
Z.~Chen, A.~Gohari, and C.~Nair, \emph{A differential equation approach
to the most-informative Boolean function conjecture}, 2025, revised 2026.
\href{https://arxiv.org/abs/2502.10019v2}{arXiv:2502.10019v2}.

\bibitem{CN24}
Z.~Chen and C.~Nair, \emph{On the optimality of dictator functions and
isoperimetric inequalities on Boolean hypercubes}, Proceedings of the
2024 IEEE International Symposium on Information Theory (ISIT),
3380--3385, 2024.
\href{https://doi.org/10.1109/ISIT57864.2024.10619180}{doi:10.1109/ISIT57864.2024.10619180}.
\href{https://chandra.ie.cuhk.edu.hk/pub/papers/HC/BF-conj.pdf}{Extended author manuscript}.

\bibitem{Chernoff52}
H.~Chernoff, \emph{A measure of asymptotic efficiency for tests of a
hypothesis based on the sum of observations},
The Annals of Mathematical Statistics \textbf{23} (1952), no.~4, 493--507.
\href{https://doi.org/10.1214/aoms/1177729330}{doi:10.1214/aoms/1177729330}.

\bibitem{CK14}
T.~A.~Courtade and G.~R.~Kumar, \emph{Which Boolean functions maximize
mutual information on noisy inputs?}, IEEE Transactions on Information
Theory \textbf{60} (2014), no.~8, 4515--4525.
\href{https://doi.org/10.1109/TIT.2014.2326877}{doi:10.1109/TIT.2014.2326877}.

\bibitem{DFR26}
P.~Durcik, M.~Fraccaroli, and J.~Roos, \emph{A weak Hellinger inequality
for noisy Boolean channels}, 2026.
\href{https://arxiv.org/abs/2609.28534v2}{arXiv:2609.28534v2}.

\bibitem{DIRX26}
P.~Durcik, P.~Ivanisvili, J.~Roos, and X.~Xie,
\emph{Sharp isoperimetric inequalities on the Hamming cube II:
The critical exponent}, 2026.
\href{https://arxiv.org/abs/2602.20462v1}{arXiv:2602.20462v1}.

\bibitem{KT26}
V.~K.~Ky and T.~Tran, \emph{Dictators are most informative}, 2026.
\href{https://arxiv.org/abs/2609.24184v1}{arXiv:2609.24184v1}.

\bibitem{LM21}
J.~Li and M.~M\'edard, \emph{Boolean functions: Noise stability,
non-interactive correlation distillation, and mutual information},
IEEE Transactions on Information Theory \textbf{67} (2021), no.~2, 778--789.
\href{https://doi.org/10.1109/TIT.2020.3041028}{doi:10.1109/TIT.2020.3041028}.
\href{https://arxiv.org/abs/1801.04462v6}{Author version, arXiv:1801.04462v6}.

\bibitem{Mirsky60}
L.~Mirsky, \emph{Symmetric gauge functions and unitarily invariant norms},
The Quarterly Journal of Mathematics \textbf{11} (1960), no.~1, 50--59.
\href{https://doi.org/10.1093/qmath/11.1.50}{doi:10.1093/qmath/11.1.50}.

\bibitem{ODonnell14}
R.~O'Donnell, \emph{Analysis of Boolean Functions}, Cambridge University
Press, Cambridge, 2014.
\href{https://doi.org/10.1017/CBO9781139814782}{doi:10.1017/CBO9781139814782}.
\href{https://arxiv.org/abs/2105.10386v1}{Corrected author edition (2021), arXiv:2105.10386v1}.

\bibitem{OpenAI26}
OpenAI, \texttt{openai/math}, repository overview, entry~119,
\emph{The Courtade--Kumar and Hellinger conjectures}.
\href{https://github.com/openai/math/blob/main/overview.tex}{Repository announcement},
accessed 7 October 2026.

\bibitem{Vallender74}
S.~S.~Vallender, \emph{Calculation of the Wasserstein distance between
probability distributions on the line}, Theory of Probability and Its
Applications \textbf{18} (1974), no.~4, 784--786.
\href{https://doi.org/10.1137/1118101}{doi:10.1137/1118101}.

\bibitem{vEH14}
T.~van Erven and P.~Harremo\"es, \emph{R\'enyi divergence and Kullback--Leibler divergence}, IEEE Transactions on Information Theory
\textbf{60} (2014), no.~7, 3797--3820.
\href{https://doi.org/10.1109/TIT.2014.2320500}{doi:10.1109/TIT.2014.2320500}.
\href{https://arxiv.org/abs/1206.2459v2}{Author version, arXiv:1206.2459v2}.

\bibitem{Yu23}
L.~Yu, \emph{On the $\Phi$-stability and related conjectures}, Probability
Theory and Related Fields \textbf{186} (2023), nos.~3--4, 1045--1080.
\href{https://doi.org/10.1007/s00440-023-01209-5}{doi:10.1007/s00440-023-01209-5}.

\bibitem{Yu26}
L.~Yu, \emph{Local optimality of dictator functions with applications to
Courtade--Kumar and Li--M\'edard conjectures}, 2024, revised 2026;
accepted for publication in The Annals of Applied Probability.
\href{https://arxiv.org/abs/2410.10147v5}{Extended version, arXiv:2410.10147v5}.

\end{thebibliography}
\end{document}